\documentclass{article}
\usepackage[utf8]{inputenc}
\usepackage{authblk}
\usepackage{setspace}
\usepackage[margin=1.25in]{geometry}
\usepackage{graphicx}
\graphicspath{ {./figures/} }
\usepackage{subcaption}
\usepackage{amsmath}
\usepackage{amsthm}
\usepackage{amsfonts}
\usepackage{amssymb}
\usepackage{graphicx}
\usepackage{verbatim}
\usepackage{breqn}

\usepackage{mathrsfs}

\newtheorem{definition}{Definition}[section]
\newtheorem{theorem}{Theorem}[section]
\newtheorem{lemma}{Lemma}[section]

\title{Simple Proof of Security of the multiparty Prepare and Measure QKD}

\author[1]{Kumar Nilesh}

\affil[1]{Department of Mathematics and Statistics, Indian Institute of Science Education and Research, Kolkata}
\date{}

\begin{document}

\maketitle

\begin{abstract}
The majority of research to date has concentrated on the quantum key distribution (QKD) between two parties. In general, the QKD protocols proposed for the multiparty scenario often involve the usage of a maximally entangled state, such as GHZ, which is challenging to implement in practice. This paper examines the prepare and measure version of multiparty communication through quantum key distribution. It is sufficient to have the capability of preparing and measuring a single qubit state. The security of the multiparty prepare and measure QKD is demonstrated by utilizing the well-known techniques of quantum error correction and entanglement purification. We begin by establishing the security of the entangled-based version and progress through a series of reductions to arrive at our equivalent multiparty prepare and measure version of QKD. We establish the security proof of multiparty entangled-based version by combining FTQC and classical statistics on the GHZ basis. Then proved the conditional security of the multiparty prepare-and-measure QKD by the equivalency.

\end{abstract}

\section{Introduction}
One of the most practical applications of quantum information science is quantum key distribution (QKD). Significant advances have been made both theoretically and experimentally since the initial QKD protocols were proposed \cite{Advances}. A QKD protocol essentially provides a systematic mechanism for two honest parties, commonly referred to as Alice and Bob, to produce a shared secret key, connected via and an authenticated classical channel and an insecure quantum channel. Because of the protocol's simplicity, it has been widely implemented and even monetized in various forms \cite{experiment}. Most QKD protocols do their security proof in a simplified scenario, in which the parties exchange an infinite amount of quantum bits (asymptotic scenario). Of course, this is far from reality, but it substantially simplifies the proof and provides insight into the protocol's real-world performance.

The majority of quantum communication theoretical and experimental work has so far concentrated on two-party protocols. Prior research has shown that Alice can concurrently communicate a key after two separate secret keys are shared through a typical two-party QKD protocol, between Alice and Bob and Alice and Charlie \cite{sudhir}. Our proposed protocol is different because it lets multiple people share a secret key at the same time. When Alice broadcasts a confidential message to two receivers, Bob and Charlie, the three persons involved must have a shared private key. This is referred to as the multiparty quantum key distribution problem. Multiparty quantum key distribution protocols seek to provide information theoretic security for multi-party communication based on quantum mechanics principles, which surpass their classical counterparts physically. The analysis of multipartite QKD protocols, on the other hand, has only been done in the implausible scenario of an endless number of signals being exchanged over the quantum channel. Nonetheless, multiparty quantum communication protocols has been demonstrated by the fascinating examples of quantum secret sharing (QSS) \cite{QSS}, quantum cryptographic conferencing (QCC) \cite{QCC}, third-man quantum cryptography \cite{TQC}, and others. Recent research has focused on the generalization of two-party protocols to multiparty schemes \cite{Xiao}. It has been demonstrated that there are certain quantum network schemes \cite{Epping} and noise regimes \cite{Ribeiro} where executing a multipartite technique is favorable in order to establish a multipartite secret key. These multiparty protocols require a critical tool: the Greenberger-Horne-Zeilinger (GHZ) entangled states that has multiparty quantum correlations that were initially developed to demonstrate the extreme violation of local realism by quantum physics \cite{GHZ, GHZ2}.

Eve is expected to do whatever manipulations are permitted by quantum mechanics on quantum systems communicated via quantum channels. This assumption is similar to Ref. \cite{assumptions2} and is a multiparty refinement of Ref. \cite{assumptions1}. A possible eavesdropper is completely unrestricted in her actions on the insecure quantum channel, as she is free to launch any type of attack called coherent attack on the quantum communication. We also assume that all parties have a genuine randomness source and that the equipment used to measure quantum systems operate in accordance with their ideal functioning.

Our approach is a generalization of the simple proof of security presented by Shor and Preskill \cite{BB84, shorpresskill}. Mayers first proved the security of BB84 \cite{mayers}. Unfortunately, this proof was rather tricky, and only a few people understood it. More work was required to genuinely comprehend, rather than simply assume, that BB84 was secure. Shor and Preskill provided a simple proof of security in 2000, in which they used the security proof of the modified Lo-Chau QKD and then reduced it to BB84 via a sequence of reductions \cite{lochau}. The initial security proof of the Modified Lo-Chau QKD was very straightforward. But, it requires non-trivial results that prove an intuitive idea that low von Neumann entropy follows from high fidelity. Further, the proof also requires the Holevo bound. The Shor Preskill proof was particularly beautiful because it combined difficult conclusions from Ref. \cite{BDSW} and Ref. \cite{lochau} to demonstrate that BB84 was secure in a simple and clear manner. We will employ a similar reduction technique to arrive at a multiparty prepare and measure QKD algorithm. To fully comprehend all of the reductions, we require an intermediate QKD known as the multiparty CSS QKD.

\subsection{BB84}
Bennett and Brassard proposed the first protocol on Quantum Key Distribution (QKD) in 1984, famously known as BB84 \cite{BB84}. Here the classical information is encoded in qubits which is a two-level quantum system.

The protocol follows the following basic structure: Alice wants to send a secret key to Bob. For this process, they have access to the classical and quantum channel. As the name suggests, the classical channel sends classical bits of information. It is further assumed that this channel should be authenticated, but any eavesdropper can listen to all the conversation. On the other hand, they use a quantum channel to send quantum signals. This channel is assumed to be completely insecure. Any eavesdropper can perform all kinds of operations on quantum signals, provided it is allowed by quantum mechanics.

This is a two-part procedure. The first part of the protocol involves quantum communication which contains the preparation, transmission, and measurement of quantum signals. The second part is entirely classical. It involves classical post-processing, in which, using classical communication, both parties use the string obtained from the first part to generate a secret key.

For over a decade, BB84 was assumed to be was secure due to the No-Cloning Theorem \cite{noclonning}. Since Eve cannot perfectly copy the qubits sent to Bob, there is little chance she shares the same result as Bob. Furthermore, measurement implies disturbance, so if Eve is too aggressive, Bob's test bits will differ from Alice's test bits, thus forcing Alice and Bob to abort. One possibility is that Eve could simply block or manipulate transmitted messages forcing Alice and Bob to abort. This is commonly known as a denial of service attack DoS \cite{dos}. A DoS attack is not a complete attack because it prevents Alice from communicating with Bob; consequently, Eve has nothing to listen to.

The original BB84 was also impractical. If any errors occurred, whether from an adversary or natural noise, the protocol would abort. Due to the prevalence of noise in nature, the original BB84 will almost always fail in practice. This led to a push to alter BB84 to make it more practical by adding classical post-processing, such as additional classical computation and/or classical communication. Adding classical post-processing to BB84 has the benefit of increasing its security without increasing its difficulty to implement. However, such additions made BB84 more complex and made it seemingly more difficult to prove its security.

Bennett and Brassard gave a very intuitive argument for the security of BB84. However, in order to prove it to be information-theoretically secure, we should relax any restriction on Eve, provided Eve is bounded by quantum mechanics. In other words, she is allowed to perform the most general attack called Coherent attack.

The first proof of security for a QKD was proposed by Lo and Chau in 1999 \cite{lochau}. Their protocol was not prepared and measure, and they utilized entanglement distillation protocol (EDP) \cite{EDP}. This proof established a link between quantum error-correcting codes and quantum key distributions. However, both parties will require large-scale quantum computers to accomplish this. In 2000, Shor and Preskill found a "simple proof of security for BB84" based on Mayers' unique use of error-correcting codes and the Lo-Chau approach to the entanglement-based version \cite{shorpresskill, mayers, lochau}. This paper proves the security of a practical multiparty QKD using the Shor-Preskill style and provides all necessary background material to understand the security proof.

In this paper, we prove the conditional security of a practical version of multiparty prepare and measure that even with the presence of acceptable noise can successfully generate secret keys. We do so by first describing and proving the security of the Entangled-Based Quantum Key Distribution (QKD) protocol. Following that, we use one of the reduction techniques devised by Shor and Preskill to reduce the Entangled-based QKD to a practical multiparty prepare and measure QKD.

\subsection{CSS Code}

Before we define quantum error-correcting codes, specifically CSS codes, we are going to describe the classical error-correcting codes. We will discuss error-correcting codes in-depth in this section and then use them for entanglement purification in the following section.

\begin{definition}
For a binary $n$ dimensional vector space $\mathbb{Z}_{2}^{n}$, the $k$ dimensional binary linear code $\mathscr{C}$ is defined as $[n,k,d]$. Where $n$ denotes the code length, $k$ is the dimension of the code and $d$ denotes the minimum distance of all codewords in $\mathscr{C}$. The rate of this linear code is given by $\frac{k}{n}$.
\end{definition}

The message space of an $[n, k, d]$ binary linear code is $\mathbb{Z}_{2}^{k}$. Note that the code can be any $k$-dimensional subspace of $\mathbb{Z}_{2}^{n}$, so there exist many $[n, k, d]$ binary linear codes for fixed parameters, $n, k$, and $d$.

\begin{definition}
The generator matrix denoted as $\mathscr{G}$ for an $[n, k, d]$ linear code $\mathscr{C}$ is defined as $k \times n$ matrix whose row space is equal to $\mathscr{C}$.
\end{definition}

It's worth noting that a code can be defined solely by its generator matrix. If $\mathscr{G}$ is a matrix with binary entries, its row space is the code generated by $\mathscr{G}$.

Let $\mathscr{G}$ be the generator matrix of $[n, k, d]$ linear code $\mathscr{C}$. To encode a message $m \in \mathbb{Z}_{2}^{k}$ into a codeword $c \in \mathscr{C}$, one simply performs the following operation:
$$
m \mathscr{G}=c .
$$

\begin{theorem}[Hamming Bound]
\cite{hamming} With the nearest neighbour decoding any $[n, k, d]$ linear code $\mathscr{C}$ can correct upto $t=\left\lfloor\frac{d-1}{2}\right\rfloor$ bit flip errors.
\end{theorem}

Another very essential matrix is the parity check matrix. This matrix is used to identify errors in codewords and thus can help in the decoding procedure.

\begin{definition}
The parity check matrix denoted as $\mathscr{H}$ for an $[n, k, d]$ linear code $\mathscr{C}$ is defined as $(n-k) \times n$ matrix with the property that
$$
\mathscr{H} c^{T}=0 \Longleftrightarrow c \in \mathscr{C}
$$
\end{definition}

Since $\mathscr{C}=\operatorname{ker} \mathscr{H}=\left\{v \in \mathbb{Z}_{2}^{n} \mid v \mathscr{H}=0\right\}$, like the generator matrix, the parity check matrix alone also defines the code $\mathscr{C}$. We use the parity check matrix to speed up decoding by calculating the syndrome, $s=\mathscr{H} b^{T}$.

\begin{theorem}
\cite{mceliece} Let a $t$ error correcting $[n, k, d]$ linear code $\mathscr{C}$ whose parity check matrix is $\mathscr{H}$. For every $e \in E=\left\{e \in \mathbb{Z}_{2}^{n} \mid \operatorname{wt}(e) \leq t\right\}$ where $t=\left\lfloor\frac{d-1}{2}\right\rfloor$, there is a unique syndrome s such that
$$
s=\mathscr{H} e^{T}
$$
\end{theorem}

To understand the CSS quantum error-correcting code, we need to define the dual code of a binary linear code. The dual code will later be used for phase error correction.

\begin{definition}
The dual code of an $[n, k, d]$ linear code $\mathscr{C}$ is defined as the vector space orthogonal to $\mathscr{C}$. Namely,
$$
\mathscr{C}^{\perp}=\left\{v \in \mathbb{Z}_{2}^{n} \mid c \cdot v=0, \forall c \in \mathscr{C}\right\}
$$
\end{definition}

The following theorem describes the properties of a dual code.

\begin{theorem}
\cite{mceliece} For an $[n, k, d]$ linear code $\mathscr{C}$ whose generator and parity check matrices are  $\mathscr{G}$ and  $\mathscr{H}$ respectively. The dual code $\mathscr{C}^{\perp}$ is a linear code $\left[n, n-k, d^{\prime}\right]$, whose generator and parity check matrices are $\mathscr{G}^{\perp}=\mathscr{H}$ and $\mathscr{H}^{\perp}=\mathscr{G}$ respectively, for appropriate $d^{\prime}$.
\end{theorem}

The quantum analog of classical noise is simply called quantum noise. Beyond classical error correction, quantum error-correcting codes can also purify damaged entangled qubits and are the foundation of proof of security of the prepare and measure QKD such as BB84.

\begin{definition}
An $[[n, k]]$ quantum error correcting code $\mathscr{Q}$ is a $k$-dimensional subspace of $\mathscr{H}^{\otimes n}$ with length $n$ and rate $\frac{k}{n}$.
\end{definition}

In this section, we will create a code that will correct against a $t$-out-of- $n$ $Z$-channel followed by a $t$-out-of- $n$ $X$-channel. The quantum version of error-correcting code that we will use is called CSS code discovered by Calderbank and Shor and independently by Steane and lies at the heart of the proof of security \cite{CSS, CSS2}.

\begin{definition}
Let $\mathscr{C}_{1}$ be a $t_{1}$-error correcting $\left[n, k_{1}, d_{1}\right]$ binary linear code, $\mathscr{C}_{2}$ an $\left[n, k_{2}, d_{2}\right]$ binary linear code, and $\mathscr{C}_{2}^{\perp}$ a $t_{2}$-error correcting binary linear code, where
$$
\mathscr{C}_{2} \subset \mathscr{C}_{1}
$$
A $t$-error correcting $[[n, k]]$ CSS code is a subspace $\mathscr{Q}$ of $\mathscr{H}^{\otimes n}$ derived from $\mathscr{C}_{1}$ and $\mathscr{C}_{2}$ with $t=\min \left(t_{1}, t_{2}\right)$ and $k=k_{1}-k_{2}$. $\mathscr{Q}$ is spanned by the CSS codewords
$$
\frac{1}{\sqrt{\left|\mathscr{C}_{2}\right|}} \sum_{c \in \mathscr{C}_{2}}|d \oplus c\rangle
$$
where $d$ runs through $\mathscr{C}_{1}$.
\end{definition}

\begin{lemma}[CSS codewords correspond to cosets] \label{lemma1}
\cite{mceliece} Let $d, d^{\prime} \in \mathscr{C}_{1}$. Then $d$ and $d^{\prime}$ belong to the same coset in $\mathscr{C}_{1} / \mathscr{C}_{2}$, i.e. $d \oplus \mathscr{C}_{2}=d^{\prime} \oplus \mathscr{C}_{2}$, if and only if
$$
\frac{1}{\sqrt{\left|\mathscr{C}_{2}\right|}} \sum_{c \in \mathscr{C}_{2}}|d \oplus c\rangle=\frac{1}{\sqrt{\left|\mathscr{C}_{2}\right|}} \sum_{c \in \mathscr{C}_{2}}\left|d^{\prime} \oplus c\right\rangle
$$
\end{lemma}

According to the following theorem, CSS codes are robust, meaning they can correct any type of qubit error as long as the errors are below a certain threshold. And perhaps the most compelling reason to use CSS code rather than any other quantum error correction code is that CSS code allows us to correct phase and bit errors in any quantum state independently. This is essential because we are not concerned with phase errors in the prepare and measure version of QKD. Additionally, if phase and bit errors can be corrected, any general type of qubit error can be fixed as well.

\begin{theorem}[CSS Codes are Robust]
\cite{mceliece} Let $\mathscr{Q}$ be a t-error correcting $[[n, k]]$ CSS code that corrects up to $t$ $X$ errors and up to $t$ $Z$ errors. Then $\mathscr{Q}$ corrects any errors on at most $t$ codeword qubits.
\end{theorem}

The definition of CSS code has already been provided. We will, however, require parametrized code. Thus, we will define parameterized CSS code with the parameters $x$ and $z$ in this section.

\begin{definition}[Parameterised CSS code]
Let $\mathbb{F}_{2}^{n}$ be a binary vector space of $n$ bits. Let $\mathscr{C}_{1}$ and $\mathscr{C}_{2}$ be two linear codes, each of which corrects upto $t$ bit errors, contained in $\mathbb{F}_{2}^{n}$ with
$$
\{0\} \subset \mathscr{C}_{1} \subset \mathscr{C}_{2} \subset \mathbb{F}_{2}^{n}
$$
Let $x, z$ be $n$ bit binary strings. Then, $\mathscr{Q}_{x, z}$ is said to be a parameterised CSS code, where for each vector in $v \in \mathscr{C}_{1}$, the respective codeword is,
$$
v \longrightarrow \frac{1}{\sqrt{\left|\mathscr{C}_{2}\right|}} \sum_{w \in \mathscr{C}_{2}}(-1)^{w \cdot z}|v+w+x\rangle
$$
\end{definition}

\section{Entanglement Purification}
In this section, we will describe the Entanglement Purification Protocol. For a mixed state that is distributed among multiple parties, this protocol is used to increase the purity using local operation and classical communication. For instance, let $n$-qubit pure entangled state $|\psi\rangle$ that is distributed between $N$ parties using a noisy quantum channel. Each of the qubits can acquire an error due to the 3 Pauli operators $\sigma_{x}, \sigma_{y}$, $\sigma_{z}$, or any combination of them \cite{knill}. Let the final state of $|\psi\rangle$ due to errors is a mixed state $\hat{\rho}$. Suppose if $m$ number of copies of $|\psi\rangle$ is sent through the noisy channel, then all parties can distill a small number of $|\psi\rangle$ with high purity. They only require to perform local operations and classical channel for this.

For the purpose of entanglement purification, we will use Calderbank, Shor, and Steane (CSS) codes \cite{CSS, CSS2}. CSS states constitute both the Bell states and GHZ states \cite{cssstate}. CSS state classes can be shown to be equivalent to the state of two colorable graphs \cite{twocolor}. Further, the purification process for CSS states are known already \cite{cssstate}.

To determine the errors uniquely, we use the parity check matrix $\mathscr{H}$ to measure the parity, provided the error is less than a predefined threshold $t$. All of these concepts were covered in the previous section. Measurements are assumed to be protective, so the final state is the projection in one of the eigenspaces of the parity check matrix. For each measurement of parity, we get the eigenvalues $\pm 1$. From this vector of eigenvalues, we obtain the syndrome $s$.

The syndrome uniquely identifies the error, allowing it to be corrected. The main idea of this protocol is that each party measure their respective parity of CSS code on their qubits. Then, by comparing their measurements to those of others, they can determine where the errors in their state are located. The errors are then corrected, and the qubits are decoded to return to the initial entangled state.

To initiate a purification protocol, the parties apply randomly selected stabiliser elements to the noisy states they wish to purify. Each side simply applies an agreed-upon single-qubit Pauli matrix to their respective state's qubit. These Pauli matrices comprise a random stabiliser element, which the parties agree upon prior to initiating the protocol via classical communication. They can therefore regard any noisy state as a probabilistic mixture of states with the same stabiliser but different eigenvalues for the generators. This was demonstrated for Bell states and was later extended to multiparty CSS stabiliser states. \cite{dur, hostens1, hostens2}.

The allowable quantity of noise per particle in GHZ states, on the other hand, decreases as the number of particles increases. In other words, as the number of particles increases, GHZ states become increasingly difficult to purify.

\subsection{GHZ State}
Greenberger, Horne, and Zeilinger first introduce the GHZ state, also known as the cat state \cite{GHZ}. The $n$ qubit state of the GHZ state is written as
$$
\left|{GHZ}^{n}\right\rangle=\frac{1}{\sqrt{2}} \overbrace{|00 \ldots 0\rangle}^{n}+\frac{1}{\sqrt{2}} \overbrace{|11 \ldots 1\rangle}^{n}=\frac{1}{\sqrt{2}}\left|0^{n}\right\rangle+\frac{1}{\sqrt{2}}\left|1^{n}\right\rangle .
$$
The GHZ state is widely used in a variety of applications, including communication complexity \cite{complexity}, nonlocality \cite{nonlocality}, and multi-party cryptography \cite{lee}, to name a few.

The three-qubit GHZ state is represented as
$$
\left|\phi^{+++}\right\rangle=\frac{|000\rangle+|111\rangle}{\sqrt{2}} \text {. }
$$
Along with the $W$ state, the GHZ state represents non-biseparable classes of 3-qubit states \cite{acin}. They cannot be transformed, even probabilistically, into one another with only local operations.

Now consider the generalization of the GHZ state. For a $M>2$ subsystem, we can create a general GHZ entangled state. If $d$ represents the dimension of each subsystem, then the overall Hilbert space of the total system will be $\mathcal{H}_{\text {tot }}=\left(\mathbb{C}^{d}\right)^{\otimes M}$. This GHZ entangled state is called $M$-partite qudit GHZ state and is denoted as
$$
|\mathrm{GHZ}\rangle=\frac{1}{\sqrt{d}} \sum_{i=0}^{d-1}|i\rangle \otimes \cdots \otimes|i\rangle=\frac{1}{\sqrt{d}}(|0\rangle \otimes \cdots \otimes|0\rangle+\cdots+|d-1\rangle \otimes \cdots \otimes|d-1\rangle)
$$
It reduces to our original state when each subsystem is of dimension 2, i.e., $M$-qubits
$$
|\mathrm{GHZ}\rangle=\frac{|0\rangle^{\otimes M}+|1\rangle^{\otimes M}}{\sqrt{2}} .
$$
In simple terms, GHZ is a maximally entangled quantum state that is in a superposition of all subsystems in state $a$ with all subsystems in $a^{\perp}$, where $a$ and $a^{\perp}$ are completely distinguishable states.

Even though the GHZ state is maximally entangled, its most intriguing property is that tracing out any single party completely destroys the entanglement and results in a fully mixed separable state.
$$
\operatorname{Tr}_{k}\left(|G H Z\rangle_{n}(G H Z \mid)=I_{n} \backslash k\right.
$$

Because entanglement monotones are maximized by the GHZ state, it is considered a maximally entangled state for the multipartite case. Additionally, it belongs to the class of states whose normal form is the same as the state itself \cite{normalform}. For the three-qubit case, there is only one such state which is nothing but the GHZ state. As a result, all three-qubit states with a non-zero normal form exhibit some sort of GHZ entanglement.

\subsection{GHZ Basis}
We can also represent any amount of entanglement with the GHZ state, similar to the bipartite case \cite{cunha}. We can begin by writing 
$$
\left|\phi^{+++}\right\rangle=\cos \theta|000\rangle+\sin \theta|111\rangle
$$
for $\theta \in (0, \pi / 2)$. With just performing operations locally on $\left|\phi^{+++}\right\rangle$ we can create a complete GHZ basis. The general structure of the states is given by
$$
\left|\phi^{\mu \lambda \omega}\right\rangle=\sum_{j}(-1)^{\mu j} b_{\mu \oplus j}|j, j \oplus \lambda, j \oplus \omega\rangle,
$$
for $b_{0}=\cos \theta$ and $b_{1}=\sin \theta$. Individually for 3 qubit state
$$
\begin{array}{ll}
\left|\phi^{+++}\right\rangle=\cos \theta|000\rangle+\sin \theta|111\rangle, &
\left|\phi^{+++}\right\rangle=\sin \theta|000\rangle-\cos \theta|111\rangle, \\
\left|\phi^{++-}\right\rangle=\cos \theta|001\rangle+\sin \theta|110\rangle, &
\left|\phi^{-+-}\right\rangle=\sin \theta|001\rangle-\cos \theta|110\rangle, \\
\left|\phi^{+-+}\right\rangle=\cos \theta|010\rangle+\sin \theta|101\rangle, &
\left|\phi^{--+}\right\rangle=\sin \theta|010\rangle-\cos \theta|101\rangle, \\
\left|\phi^{+--}\right\rangle=\cos \theta|011\rangle+\sin \theta|100\rangle, &
\left|\phi^{---}\right\rangle=\sin \theta|011\rangle-\cos \theta|100\rangle
\end{array}
$$

${}$

\begin{table}
\centering
\begin{tabular}{c|c|ccc}
\hline & Computational basis & $X_{1} X_{2} X_{3}$ & $Z_{1} Z_{2}$ & $Z_{2} Z_{3}$ \\
\hline$\left|\phi^{+++}\right\rangle$ & $(|000\rangle+|111\rangle) / \sqrt{2}$ & $+1$ & $+1$ & $+1$ \\
$\left|\phi^{++-}\right\rangle$ & $(|001\rangle+|110\rangle) / \sqrt{2}$ & $+1$ & $+1$ & $-1$ \\
$\left|\phi^{+-+}\right\rangle$ & $(|011\rangle+|100\rangle) / \sqrt{2}$ & $+1$ & $-1$ & $+1$ \\
$\left|\phi^{+--}\right\rangle$ & $(|010\rangle+|101\rangle) / \sqrt{2}$ & $+1$ & $-1$ & $-1$ \\
$\left|\phi^{-++}\right\rangle$ & $(|000\rangle-|111\rangle) / \sqrt{2}$ & $-1$ & $+1$ & $+1$ \\
$\left|\phi^{-+-}\right\rangle$ & $(|001\rangle-|110\rangle) / \sqrt{2}$ & $-1$ & $+1$ & $-1$ \\
$\left|\phi^{--+}\right\rangle$ & $(|011\rangle-|100\rangle) / \sqrt{2}$ & $-1$ & $-1$ & $+1$ \\
$\left|\phi^{---}\right\rangle$ & $(|010\rangle-|101\rangle) / \sqrt{2}$ & $-1$ & $-1$ & $-1$ \\
\hline
\end{tabular}
\caption{The three qubit GHZ basis states $\left|\phi^{s_{1} s_{2} s_{3}}\right\rangle$ and its representation in shown in the table, along with the stablizers of the corresponding states.Here the relations between the generator $s_{1} X_{1} X_{2} X_{3}, s_{2} Z_{1} Z_{2}$, $s_{3} Z_{2} Z_{3}$. and the basis is denoted by the the sign $s_{1}, s_{2}$, and $s_{3}$.}
\label{table:1}
\end{table}

\section{Entangled-based version}
Our entanglement-based version of multiparty QKD employs a correlation similar to that of the Ekert protocol, i.e., a perfect correlation is obtained when a maximally entangled state is measured in the same direction \cite{E91}. For simplicity of notation we will denote $\left|GHZ\right\rangle$ state as $\left|\Phi^{+}\right\rangle$. Suppose that all parties share the maximally entangled GHZ state
$$
\left|\Phi^{+}\right\rangle=\frac{1}{\sqrt{2}}\left(|0 \cdots 0\rangle+|1 \cdots 1\rangle\right)
$$
This state cannot be entangled with any other state because it is a pure state. So an eavesdropper cannot obtain any information about this state. Hence, the aim for all parties is to share $m$ of these maximally entangled state
$$
\left|\Phi^{+}\right\rangle^{\otimes m}=\left|\Phi^{+}\right\rangle \otimes \cdots \otimes\left|\Phi^{+}\right\rangle
$$
and measure these states to get a common secret shared key of which Eve does not have any knowledge. However, due to the insecure nature of the quantum channel, they are using Eve can interact with these states. Further, because the channel is also noisy, the final state shared between them will not be the same state given above but a mixed state $\hat{\rho}$. They need to apply entanglement distillation in order to obtain a small number of maximally entangled states.

The step by step description of the entangled-based version of QKD can be written as follows:

\begin{enumerate}
\item From the state $\left|\Phi^{+}\right\rangle$, Alice creates $2n$ state in $\left|\Phi^{+}\right\rangle^{\otimes 2 n}$. So a total of $2 N n$ qubits in the form of $2 n$ entangled pairs is shared between $N$ parties.
\item $n$ of these states are randomly selected by her for the error estimate in the later stage. For this, she uses a permutation $P$.
\item She selects a binary string randomly $b=\left(b_{1}, b_{2}, \ldots, b_{2 n}\right)$ of length $2 n$, and if $b_j$ is 1, she performs Hadamard operation to the other qubits for the $j^{\text {th }}$ state. That is, if $j^{\text {th }}$ bit of $b$ is 1, then Hadamard operator is applied on other qubits of $j^{\text {th}}$ state. So the final $j^{t h}$ state would be,
$$
\Phi_{j}^{+}=(I \otimes H^{\otimes {N-1}}) \frac{|0\cdots 0\rangle+|1\cdots 1\rangle}{\sqrt{2}}=\frac{|0\rangle|+\rangle \cdots |+\rangle +|1\rangle|-\rangle \cdots |-\rangle }{\sqrt{2}}
$$
\item She sends all other qubit pairs to different parties.
\item After receiving their respective qubits, all parties announce this fact publicly. Finally, all the parties have $2n$ qubits each.

Alice permutes her qubits with $U_{P}$ and measures the last half of her qubits in the computational basis to estimate the error rate. The measurement outcome is stored in the bit string $w^{(A)}$.

\item Alice sends $P, b$, and $w^{(A)}$ to all other parties.

\item All other parties apply $U_{P}$ and corresponding Hadamard according to the string $b$ to their qubits and measure the last half of their qubits (the test qubits) in the computational basis. The measurement outcome is stored in the bit string $w^{(B)}_i$.

\item Everyone shares the value of $w^{(B)}_i$.

All parties calculate the total error $w$ using $w^{(B)}_i$ and $w^{(A)}$ in order to find the total number of errors.

Based on some small mutually chosen confidence factor $c \geq 0$, they calculate $t=\mathrm{wt}(w)+c n-1$.

Everyone agrees on a $t^{\prime}$ quantum error correcting $[[n, k]]$ code $\mathcal{Q}$ for some $t^{\prime} \geq t$ and $k$. If no such code exists, then they abort.

\item Error correction codes $\mathscr{Q} ( \mathscr{C}_{1}, \mathscr{C}_{2} )$ are used to correct all the errors in the remaining qubits, provided total error is below $t$. They finally obtain $\left|\Phi^{+}\right\rangle^{\otimes m}$.

Independently, all of them unencode their remaining qubits with the decoding operation $U_{\text {encode }}^{\dagger}$ for $\mathscr{Q}$.

\item Finally, after measuring in the computational basis of the state $\left|\Phi^{+}\right\rangle^{\otimes m}$, they obtain the shared secret key $k$.

\end{enumerate}

\section{Security}

For the entangled-based version of QKD, we want to show that the knowledge Eve obtains is exponentially small if the fidelity of the shares state with the GHZ state is very high. As we will show that entanglement-based protocol is equivalent to prepare and measure version, it also proves the security of prepare and measure protocol.

The complete security of the entangled-based version will be done in two steps. The first step is to reduce the noise from all the sources to get a noiseless quantum protocol. Quantum repeaters and Fault-Tolerant Quantum Computation (FTQC) are used for this step \cite{repeater, FTQC}. These concepts have been widely used in the context of QKD. We will also prove that the error syndrome shared publicly does not provide any information to Eve. Hence Eve can have complete control over the quantum repeaters and error correction information without compromising security.

Once we have reduced to the noiseless protocol, all parties still need to check if any of the qubits are transformed by the eavesdropper. For this, we will use classical statistics. This is the second step for the security proof, i.e., reducing the noiseless quantum protocol to classical verification protocol. So that we can use the classical statistics to prove the security of entangled-based multiparty QKD. The use of classical statistics for this context helps simplify our calculations significantly. In the following theorem, we demonstrate why the classical arguments work in our quantum context.

\begin{theorem}
The initial state of the quantum system can be reduced to a classical mixture, and classical statistics can be applied.
\end{theorem}
\begin{proof}
In the quantum problem, we can apply the classical statistics because the observables $\mathscr{A}_{i}$'s that we are considering are all diagonal with respect to a basis $\mathbb{B}$, namely the GHZ basis. 

Consider a complete quantum measurement operator in this basis to be represented as $\mathscr{M}$. Now since all the observables $\mathscr{A}_{i}$ that we will use are diagonal in the basis representation of $\mathbb{B}$, they all commute with any other operator, in particular with the measurement $\mathscr{M}$. Hence taking the measurement in the basis $\mathbb{B}$ will not affect the actions of any $\mathscr{A}_{i}$. So we can always assume that such a quantum measurement $\mathscr{M}$ is done prior to the subsequent actions of $\mathscr{A}_{i}$'s.

The initial state of the quantum system, in other words, is just a classical mixture of eigenstates of $\mathscr{M}$. Therefore, we can easily apply an appropriate classical argument in this quantum case.
\end{proof}

Like in the original proof given by Lo-Chau, the observables $O_{i}$'s we considered have degenerate eigenvalues, i.e., coarse-grained observables, not the fine-grained ones \cite{lochau}.

Errors in the GHZ state distribution can not only just come from the noisy channel but can also be because of imperfect sources, errors in the storage of quantum states, and errors during the computation through the elementary gates and measurement. These errors can be easily corrected using quantum error-correcting techniques and Fault Tolerant Quantum Computing. The main idea of FTQC is that the quantum states are encoded in error correction code, and then quantum operations are performed to the encoded state \cite{FTQC}. The use of quantum repeater is not necessary but will aid the error-correcting schemes if the distance is considerable. As for the distance larger than the coherence length, the errors are very high to apply the standard entanglement purification techniques.

We should note that we can extend the standard threshold result to distributed computation on a noisy quantum channel using the above techniques. That is, the GHZ distribution scheme or any distributed quantum algorithm can be easily extended to a noisy quantum scheme.

We have assumed the general case where Eve has control over the quantum channel. Eve can easily learn the information about the error syndrome obtained using the error-correcting codes. As we have already mentioned that this information does not provide any advantage to Eve. We will now explicitly prove this statement in the following theorem.

\begin{theorem}
The error syndrome calculated and shared during the error correction step provides no information about the encoded key.
\end{theorem}

\begin{proof}The state shared between all the parties can be written mathematically as a tensor product of logical qubits and the ancillary qubits. The logical qubits keep the encoded information, and the ancillary qubits are there to store the error syndrome. In other words the state $|\Phi \rangle$ can be represented as $\sum_{i, j} a_{i j}\left|\alpha_{i}\right\rangle \otimes\left|\beta_{j}\right\rangle$, where the logical qubits are represented by $\alpha$'s and the ancillary qubits are represented by $\beta$'s and $a_{i j}$ are complex coefficients. We are making no judgment on the state of the system shared between them; they can very well be entangled with Eve's system. This will, however, not affect the proof of this theorem.

With the action of FTQC, the ancillary qubits transform unpredictably

but with a very high fidelity of about $1-O\left(e^{-r}\right)$ for some selected positive $r$, the logical qubits will follow specific evolution and will not be affected by the error calculation. So as long as the errors are below the specified threshold, all the quantum operations can be assumed to be operated only on the state of the logical qubits. The ancillary qubits, in other words, decouple from the verification step.

The logical qubits will therefore represent the desired state, provided the states are not affected too much due to eavesdropping. Hence it does not matter how much information about the error syndrome Eve gains. In general the error syndrome obtained during the quantum error correction step contains no information about the encoded quantum information and can be publicly shared.

\end{proof}

Once we are done with error correction and entanglement distillation, we will be focusing on a noiseless scheme. In order to verify that the actual GHZ state is shared between the parties, we are going to use the random-hashing technique of BDSW \cite{BDSW}. The idea is an extension of the classical verification scheme. It can be easily understood with the help of a simple two-party game. Suppose Eve is hiding an $N$ bit binary string and Alice's objective is to find if all of Eve's bits are 1 or not. For this, Alice is free to ask a few number of fair questions. For example, she could ask, Are the first and third bits the same or not? Then based on Eve's answer, Alice needs to decide whether to accept the string or not. Eve will win if Alice accepts but not all of Eve's string is 1. And Alice will win if she accepts and all the strings are 1. In all other cases, the game will be considered to be a draw. The best strategy for Alice is not to ask a single-digit question but ask for the parities of random subset of bits. For instance, 'is the $i$ th bit 1' is not a good question, but 'are $i$ th and $j$ th bits same' is a good question.

Consider the following scenario: Eve is hiding a 7 digit binary string $x=1101001$ and Alice wants to know if the fourth and fifth bits are the same or not. In other words, she can ask this question as index string $s=0001100$. The answer will be given through the parity $x \cdot s ~(\bmod 7)$. This test reveals whether these two bits are the same or not. There is an exponentially small probability that Eve can cheat in this verification scheme. For the case of two parties, it was already shown that after $m$ number of such questions, the probability for Eve to cheat is $2^{-m}$ \cite{lochau}. The only criteria is that Eve should not know the parity index question, otherwise she can easily win.

Using the above classical verification scheme we can create a quantum verification scheme. We are going to consider GHZ basis and the convention used in BDSW. In particular consider the GHZ basis vectors represented in BDSW notation as follows (for illustration we have condidered only for 3 qubits here):

\begin{table}
\centering
\begin{tabular}{c|c|ccc}
\hline & Computational basis & BDSW Notation \\
\hline$\left|\phi^{+++}\right\rangle$ & $(|000\rangle+|111\rangle) / \sqrt{2}$ & $111$ \\
$\left|\phi^{++-}\right\rangle$ & $(|001\rangle+|110\rangle) / \sqrt{2}$ & $110$ \\
$\left|\phi^{+-+}\right\rangle$ & $(|011\rangle+|100\rangle) / \sqrt{2}$ & $101$ \\
$\left|\phi^{+--}\right\rangle$ & $(|010\rangle+|101\rangle) / \sqrt{2}$ & $100$ \\
$\left|\phi^{-++}\right\rangle$ & $(|000\rangle-|111\rangle) / \sqrt{2}$ & $011$ \\
$\left|\phi^{-+-}\right\rangle$ & $(|001\rangle-|110\rangle) / \sqrt{2}$ & $010$ \\
$\left|\phi^{--+}\right\rangle$ & $(|011\rangle-|100\rangle) / \sqrt{2}$ & $001$ \\
$\left|\phi^{---}\right\rangle$ & $(|010\rangle-|101\rangle) / \sqrt{2}$ & $000$ \\
\hline
\end{tabular}
\caption{The extended BDSW notation for three-qubit GHZ basis.}
\label{table:2}
\end{table}

An $m$ GHZ basis is a complete $m$ ordered basis, each of them is described by a string of length $m$ \cite{bone}. In our protocol say $n$ GHZ states are shared between all parties, then the complete state of the system can be described by string with $\tilde{1}$ 's of length $mn$, $|\tilde{1} \tilde{1} \ldots \tilde{1}\rangle$.  

We now consider the case when an eavesdropper is present. For more generality, we don't only allow Eve to act on the state shared between each member but also to create the state and then send it to Alice. So the state can also be entangled with Eve's system. To represent this general state, we consider a bigger Hilbert space. We know that there is a larger purifying Hilbert space for any mixed state such that it can be represented as a pure state. So without loss of generality, let us assume that Eve creates and provides a pure state represented as:
$$
|\psi \rangle=\sum_{i_{1}, i_{2}, \cdots, i_{N}} \sum_{j} \alpha_{i_{1}, i_{2}, \cdots, i_{N}, j}\left|i_{1}, i_{2}, \cdots, i_{N}\right\rangle \otimes|j\rangle
$$
where $i_{k}$ represents BDSW notation of $k$ th state and takes values $\tilde{0} \tilde{0} \tilde{0}$ to $\tilde{1} \tilde{1} \tilde{1}$ (for 3 parties), $\alpha_{i_{1}, i_{2}, \cdots, i_{N}, j}$'s are the coefficients, and the $|j\rangle$ represents the ancilla's orthonormal basis. Because the state is prepared by Eve herself, each such $|\psi \rangle$ can represent a different cheating strategy of Eve.

We aim to use the same technique as the classical verification scheme using this notation. Using the random-hashing idea of BDSW, we can create an effective quantum verification scheme. Like the classical case, we can find the parity from the outcomes of a particular measurement operator-defined below.

Consider Eve makes a classical mixture of product states of GHZ states. It can be easily shown that the classical verification scheme with the help of BDSW notation can be applied, and the probability of Eve cheating is exponentially small. However, because Eve is free to create any general state. We are interested in knowing if creating an arbitrary general state provides any advantage to Eve. The main step in proving the security of the Entangled-based version is to show that creating a general state in place of just a classical mixture of GHZ states does not enhance Eve's probability of cheating. This was first observed by Lo-Chau \cite{lochau}.

\begin{theorem}
The probability of Eve to cheat successfully for the case when she prepares any arbitrary general state is the same as the case when she prepares a state as a classical mixture of $n$-GHZ basis.
\end{theorem}

\begin{proof}
We define three operators on an arbitrary state $|\Psi \rangle$ of $n$ qubits shared between $N$ parties. All these operators are defined based on their operations on the complete basis formed by $n$-GHZ states.

The first operator $\mathscr{R}$ is defined on the $nN$ bits of string $r$ as follows 
$$
\mathscr{R} \quad: \quad \sum_{r} r|r\rangle\langle r|
$$
where $r$ represents the shared quantum state in the BDSW notation defined previously. The second operator $\mathscr{S}$ for any index string $s$ is defined as
$$
\mathscr{S} \quad: \quad \sum_{r}(s \cdot r)|r\rangle\langle r|
$$
This operator provides the parity for bits in the subset $s$. Finally, we define
$$
\mathscr{P} \quad : \quad |\tilde{1} \tilde{1} \ldots \tilde{1}\rangle\langle\tilde{1} \ldots \tilde{1} \tilde{1}|
$$
represents the projector operator on the maximally entangled $|\Phi^{+} \rangle$ GHZ state. Note that all the observables are defined in a single basis, which is $n$-GHZ basis. Clearly, above operators in this basis are simultaneously diagonalizable. Hence the operator $\mathscr{P}$ and all the $\mathscr{S}$'s commute with the measurement $\mathscr{R}$. Therefore the outcomes of $\mathscr{P}$ and any $\mathscr{S}$ will not be affected by any prior measurement of $\mathscr{R}$. In particular, Eve may have provided any arbitrary state $|\Psi \rangle$, all the parities calculated during the verification process and the hypothetical measurement of the projection $\mathscr{P}$ will not be affected if we measure the state $|\Psi \rangle$ in $n$-GHZ basis, which is $\mathscr{R}$, before giving to Alice in order to perform these operations.

\end{proof}

The above theorem shows that creating any arbitrary state does not give Eve any more advantage than creating a mixture of products of GHZ basis states. Hence using the quantum verification scheme, we can say that Eve's chance of cheating is exponentially small and the entangled-based version is secure. It should be noted, however, that as in the classical case, Eve should not know the parity index string ahead of time. This concludes the security proof of the multiparty Entangled-based QKD. Now we will perform a series of reductions to get our prepare and measure version.

\section{Reduction 1}


Now in order to prove that the original prepare and measure version of QKD is secure, we need to show that the prepare and measure version is equivalent to the entangled-based QKD. We can use the Entangled-based QKD as a starting point to create another QKD by carefully substituting quantum operations for classical operations or simply discarding some unnecessary quantum operations. We systematically modify each step starting with the entangled-based version to reach the original protocol finally. For this reduction, we will use various concepts from classical linear codes, CSS codes and its properties of cosets.

We first remove the requirement to share the maximally entangled state $\left|\Phi^{+}\right\rangle^{\otimes 2 n}$. We can easily remove half of them because in step 7, $n$ of them was measured to obtain the error rate. For this purpose, we do not need to use the maximally entangled state; Alice could simply send single-qubit states in the computational basis. With this change, we have modified the following steps:

$1^{\prime}$. $n$ arbitrary single qubits encoded in $|0\rangle$ or $|1\rangle$ and $n$ maximally entangled states $\left|\Phi^{+}\right\rangle^{\otimes n}$ are prepared by Alice.

$2^{\prime}$. She selects $n$ random positions for the check qubits and in the remaining $n$ positions, she puts half of the state $\left|\Phi^{+}\right\rangle$.

$7^{\prime}$. To estimate the error all the check qubits are measured in $|0\rangle,|1\rangle$ basis. If error is more than the threshold $t$ then the protocol is aborted.

Now to remove the requirement of the rest $n$ maximally entangled pairs, we will use CSS code to encode the key. Further, the following points explain how we can alter the entangled-based QKD without affecting the security.

1. At the end of the Entangled-based Protocol, everyone measures in the computational basis. As long as the number of phase errors are at most $t$, they do not use nor care about the phases or phase errors in the error-correcting codewords. This point is an obvious but very important-statement first noted by Shor and Preskill \cite{shorpresskill}. Alice and Bob shall now ignore phase error detection, correction and unencoding.

2. We share some imperfect state $\left(\Phi^{+}\right)^{\otimes n}$ which is then used to get almost perfect state $\left(\Phi^{+}\right)^{\otimes m}$ where $(m<n)$ using entanglement distillation. Equivalently it is CSS code $\mathscr{Q}$ of $n$ qubits which is protecting $m$ qubits in the noisy channel.

3. Eve's presence is ruled out with entanglement distillation. In this step, each party measures their respective syndrome to find the error. Alice could as well measure her syndrome before or after sharing, provided the errors are below a threshold. In case she chooses to measure her syndrome before sharing; with $x$ and $z$ as her error measurements for bit and phase error. If she detects $x$ and $z$ error with the first qubit of GHZ state, then all other members should also correct the same $x$ and $z$ and in addition, they also need to correct the error due to channel. Basically, for this case, the step is equivalent to sharing GHZ state encoded in parameterized CSS code $\mathscr{Q}_{x, z}$, where $x$ and $z$ are determined by Alice's syndrome measurement and is entirely random.

4. Alice could as well measure the $m$ qubits of GHZ states before sharing the state with others along with the syndrome measurement. The outcome will be a random key $k$ for everyone. Since the errors is below the threshold so CSS code is able to correct the qubits and all parties receive the qubits in the states in which Alice sent them. And since the state received by everyone is in the original state, the no-cloning theorem provides the security. If Alice measures her qubits before sending, others will receive their states with the same outcome.

From points 3 and point 4 above, this is equivalent to Alice sending a random key $k$ of length $m$ encoded in parametrized CSS code $\mathscr{Q}_{x, z}$, for $x$ and $z$ selected randomly by Alice.
$$
k \rightarrow \frac{1}{\sqrt{\left|\mathscr{C}_{2}\right|}} \sum_{w \in \mathscr{C}_{2}}(-1)^{w \cdot z}|k+w+x\rangle
$$

One of the most important features of CSS code is the codewords defined by CSS code $\left|v, x, z\right\rangle$ form an orthonormal basis of a Hilbert space of dimension $2^n$ \cite{CSS}. We can use this orthonormal basis for our GHZ state in place of the computational basis. So instead of sharing the GHZ state
$$
\left|\Phi^{+}\right\rangle^{\otimes n}=\frac{1}{\sqrt{d}} \sum_{i=0}^{d-1}|i\rangle \otimes \cdots \otimes|i\rangle
$$
They share the following equivalent GHZ state
$$
\left|\Phi^{+}\right\rangle^{\otimes n}=\frac{1}{\sqrt{d}} \sum_{u, x, z}\left|u, x, z\right\rangle \otimes \cdots \otimes \left|u, x, z\right\rangle,
$$
where $\left|u, v, w\right\rangle = \frac{1}{\sqrt{\left|\mathscr{C}_{2}\right|}} \sum_{w \in \mathscr{C}_{2}}(-1)^{w \cdot z}|k+w+x\rangle$.

On this state when Alice performs the syndrome measurement using $\mathscr{H}_{1}$ and $\mathscr{H}_{2}^{\perp}$ on her part of the system, she gets a random values for $x$ and $z$. Likewise she gets a random string $u \in$ $\mathscr{C}_{1} / \mathscr{C}_{2}$ when she measures in computational basis in the step 10. Hence all other qubits also collapsed into the codeword for $u$ in $\mathscr{Q}_{x, z}\left(\mathscr{C}_{1}, \mathscr{C}_{2}\right)$, that is $\left|u, x, z\right\rangle$. Which is nothing but the encoded part of the key state $|k\rangle$ formed from $m$ qubits. So from the point mentioned above, instead of sending the maximally entangled state, Alice could send just the $m$ qubit key state $|k\rangle$ encoded in $n$ qubit $\mathscr{Q}_{x, z}\left(\mathscr{C}_{1}, \mathscr{C}_{2}\right)$ code and two random $n$ bit strings $x$ and $z$. At this stage of reduction we reach a protocol we call CSS multiparty protocol.

The main points of the above reduction can be understood intuitively. Provided errors are below a threshold, with CSS code we can share the encoded information with really high fidelity and hence the no-cloning theorem says that very small information can leak out to eavesdropper.

\section{Multiparty CSS Protocol}
Here we present the protocol known as CSS multiparty protocol after integrating all the changes discussed in the preceding section into the entangled-based protocol. Steps of CSS multiparty protocol equivalent to entangled-based is presented below.

\begin{enumerate}
\item Following random binary strings are selected by Alice: an $n$ bit check string in the computational basis, $m$ bit key $k$ string and two $n$ bit strings $x$ and $z$.

\item The key bits $|k\rangle$ is encoded using the CSS code $\mathscr{Q}_{x, z}\left(\mathscr{C}_{1}, \mathscr{C}_{2}\right)$
$$
|k\rangle \rightarrow \frac{1}{\sqrt{\left|\mathscr{C}_{2}\right|}} \sum_{w \in \mathscr{C}_{2}}(-1)^{w \cdot z}|k+w+x\rangle
$$
She then randomly mixes the encoded key bits and the check qubits 

in $2n$ positions.

\item Based on an arbitrary classical string of length $2n$, $b=\left(b_{1}, b_{2}, \ldots, b_{2 n}\right)$, she performs Hadamard operation to all the qubits in $i$th position of the sequence if $b_{i}=1$.

\item Alice sends all the qubits to everyone.

\item All parties confirm the receiving of the qubits.

\item Once everyone received the qubits, Alice broadcasts the location of check qubits and the value of strings $b, x$, and $z$.

\item Other parties perform Hadamard operation whenever the value of $b_{i}$ is 1.

\item All of the check qubits are then measured in $\{|0\rangle, |1\rangle \}$ basis. If the error is more than the threshold $t$, the protocol is aborted.

\item Otherwise everyone performs the error correction and decodes the left $m$ qubits from $\mathscr{Q}_{x, z}$ $\left(\mathscr{C}_{1}, \mathscr{C}_{2}\right)$.

\item They finally perform the measurement to obtain a shared secret key $k$.
\end{enumerate}

From the entangled-based version, which was proved to be secure, we have reduced to reach the CSS protocol. By equivalency, this protocol is also secure. However, we have yet to reach at the prepare and measure version of the protocol. There is still the requirement of large enough quantum computers to encode and decode the key. Further, all other parties require quantum memory to store their qubits until Alice announces the binary strings $b, x$, and $z$. On the other hand, the prepare and measure only need the ability to prepare and measure single qubits. Because the CSS code decouples the bit error from the phase error, it is easy to relax all these requirements.

\section{Reduction 2}

In this section, we will look into the similarities and differences between multiparty CSS and Prepare and Measure version. Then step by step, we will reduce the CSS code protocol to get our final multiparty prepare and measure QKD.

Previously Alice broadcasted the binary strings $x$ and $z$ with others to perform the entanglement purification. However, the key is determined solely by the qubit's value, not by entanglement. In other words, the key is generated through the value of the qubits, and the phase is not needed. So they don't need to correct the phase now. Hence only string $x$ is required in order to correct the bit error and Alice can discard the value of $z$. If we ignore the value of $z$, the effective state shared by Alice is a mixture of $|k+x+w\rangle$ (as described in the theorem below). Because Alice does not share $z$ the effective state shared by her is a mixed state obtained by averaging over the values of $z$. This is clearly proved in the following theorem.

\begin{theorem}
In the multiparty CSS protocol the effective mixed state shared by Alice is
$$
M=\frac{1}{\left|\mathscr{C}_{2}\right|} \sum_{w \in \mathscr{C}_{2}}\left|k^{\prime}+w+x\right\rangle\left\langle k^{\prime}+w+x\right|
$$
with other parties when taking average over the phase syndrome $z$.
\end{theorem}

\begin{proof}
Consider a binary string $x$ of length $n$. We define $s(x)\subset\{0, 1, \dots, n-1\}$ to be the set denoting the locations of $1$ in the string $x$. The main point of the argument is that when $z$ has an even number of $1$ on the positions in $s(x)$ then the value of the dot product $x \cdot z$ is zero. And when $z$ has odd number of $1$ on the positions in $s(x)$ then the value of the dot product $x \cdot z$ is one. We can represent it as
$$
x \cdot z = |s(x)\cup s(z)|\mod 2.
$$
Equivalently
$$
(-1)^{x \cdot z} = \begin{cases}
+1&\text{for}\quad|s(x)\cup s(z)|\quad\text{is even}\\
-1&\text{for}\quad|s(x)\cup s(z)|\quad\text{is odd}.
\end{cases}
$$
Hence,
$$
\sum_z (-1)^{x\cdot z}=\begin{cases}
2^n&\text{for}\,x=0\\
0&\text{otherwise}
\end{cases}
$$
Here the sum is taken over all $n$ bit string. It is easy to see if we consider two different cases. $s(x) = \emptyset$ for $x=0$ and each of the $2^n$ terms in the sum is $+1$. And for the case when $x \neq 0$ then half of the terms in the sum correspond to $z$ with an odd number of $1$ on positions in $s(x)$ and half of the terms correspond to $z$ with an even number of $1$ on positions in $s(x)$. Hence the terms get canceled and results in zero.

Now coming to our original problem, we need to sum the state over the values of $z$.
$$
\begin{aligned}
\frac{1}{2^{n}\left|C_{2}\right|} \sum_{z}\left[\sum_{w_{1}, w_{2} \in C_{2}} (-1)^{\left(w_{1}+w_{2}\right) \cdot z} \times\left|k^{\prime}+w_{1}+x\right\rangle\left\langle k^{\prime}+w_{2}+x\right|\right]
\end{aligned}
$$
We can break the equation into two different cases $w_1=w_2$ and $w_1\ne w_2$ and then use the previous equation. We have
$$
\frac{1}{2^n|\mathscr{C}_2|}\sum_z\sum_{w_1,w_2\in \mathscr{C}_2}(-1)^{(w_1+w_2)\cdot z}|k'+w_1+x\rangle\langle k'+w_2+x|
$$
$$
= \frac{1}{2^n|\mathscr{C}_2|}\sum_{w_1,w_2\in \mathscr{C}_2}\left(\sum_z(-1)^{(w_1+w_2)\cdot z}\right)|k'+w_1+x\rangle\langle k'+w_2+x|
$$
$$
= \frac{1}{2^n|\mathscr{C}_2|}\sum_{w_1,w_2\in \mathscr{C}_2\\w_1\ne w_2}\left(\sum_z(-1)^{(w_1+w_2)\cdot z}\right)|k'+w_1+x\rangle\langle k'+w_2+x| 
$$
$$
+ \frac{1}{2^n|\mathscr{C}_2|}\sum_{w\in \mathscr{C}_2}\left(\sum_z(-1)^{0\cdot z}\right)|k'+w+x\rangle\langle k'+w+x|
$$
$$
= 0 + \frac{1}{2^n|\mathscr{C}_2|}\sum_{w\in \mathscr{C}_2}2^n|k'+w+x\rangle\langle k'+w+x| 
$$
$$
= \frac{1}{|\mathscr{C}_2|}\sum_{w\in \mathscr{C}_2}|k'+w+x\rangle\langle k'+w+x|
$$
as stated in the theorem.
\end{proof}


Using the expression from the above theorem, we can say that equivalently it is the mixture of the states $|k+w+x\rangle$, with $w$ chosen at random in $\mathscr{C}_{2}$. By picking a codeword, this state can be created entirely utilising single-qubit operations. As a result, the initial few steps of the multiparty CSS protocol can be adjusted as follows:

Alice creates following bits: random $n$ check bits, codeword $u$ selected at random from $\mathscr{C}_{1} / \mathscr{C}_{2}$, another codeword $y \in \mathscr{C}_{2}$, and a random string $w$ of length $n$. She encodes $n$ qubits in the state $|0\rangle$ or $|1\rangle$ according to $u+x+w$, and $n$ qubits in $|0\rangle$ or $|1\rangle$ according to the check bits. Others will receive a state $\left|u+x+w+e_i\right\rangle$ and they directly measure this state in the $|0\rangle,|1\rangle$ basis, instead of first decoding the state and then measuring it as described in steps 9 and 10. Then they obtain a string $u+x+w+e_i$ and now they can perform the error correction classically because $\mathscr{C}_{1}$ and $\mathscr{C}_{2}$ are classical error correcting codes. With the error-correcting information $x$ they have received from Alice, they can subtract $x$ from his string and obtain $u+w+e_i$. If $e_i$ does not contain too many errors, others can then correct $u+w+e_i$ unambiguously to the codeword $u+w$.

This procedure can be summarized as follows:
$$
|\mathbf{k}+\mathbf{w}+\mathbf{x}\rangle \xrightarrow{\text{send}} \quad|\mathbf{k}+\mathbf{w}+\mathbf{x}+\mathbf{e_i}\rangle \xrightarrow{\text{measure}} \mathbf{k}+\mathbf{w}+\mathbf{x}+\mathbf{e_i} \xrightarrow{\text{- x}} \mathbf{k}+\mathbf{w}+\mathbf{e_i} \xrightarrow[\text{correct}]{\text{error}} \mathbf{k}+\mathbf{w} \xrightarrow{\text{decode}} \mathbf{k}
$$

We have successfully modified the protocol such that the key qubits are prepared in a random state, similar to that of check bits. For simplicity and to appear more like our traditional prepare and measure version, we have modified some variables below. Using all the above arguments, certain steps of the protocol can be modified to:

$1^{\prime \prime}$. Random $|u \rangle  \in \mathscr{C}_{1}$ is selected by Alice. Further she creates $n$ random qubits in the computational basis.

$2^{\prime \prime}$. She randomly mixes the check qubits and the encoded qubits $|u\rangle$ in these $2n$ positions.

$6^{\prime \prime}$. Alice announces $b, u+v$, and the positions of the $n$ check bits, where $v$ is the bits after the shifting step.

$9^{\prime \prime}$. Others measure the remaining qubits and shift to obtain $v+e_i$, subtract $u+v$ from this, and correct it with the code $u \in C_{1}$ to obtain $u$.

$10^{\prime \prime}$. All parties compute the coset from which $u$ belongs to get the final key $k$.

Note that the step $10^{\prime \prime}$ is justified because of the lemma \ref{lemma1}. Finally, we can also remove the requirement of Hadamard operation. First, note that in multiparty prepare and measure QKD, when others receive the qubits, they measure them randomly either in computational basis $|0\rangle,|1\rangle$ or in Hadamard basis $|+\rangle,|-\rangle$. They then discard all the qubits that were not measured in the same basis. Alice could simply select a random binary string $b$ and based on this, she can directly prepare her qubits in either computational basis $|0\rangle,|1\rangle$ or in Hadamard basis $|+\rangle,|-\rangle$. All other members can also directly measure these qubits randomly in these two bases. And finally, they only keep those qubits for which the chosen basis matches for everyone. This further removes the need for quantum memory to keep the qubits till they hear the information from Alice. However, since there are $N$ parties and they only keep those qubits for which the chosen basis is the same, they need to begin with $n(1+2^N)$ qubits to get the final key of the same length. Note that this is still significantly less than performing individual QKD between $N$ parties, for which we require at least $4^{N}n$ qubits.

Now Alice needs to wait till the irrelevant qubits are discarded before broadcasting the location of check bits. With these reductions, we finally reach our multiparty prepare and measure QKD, with additional error correction and privacy amplification steps.

\section{Secure multiparty Prepare and Measure Version}

\textbf{Secure prepare and measure QKD}

\begin{enumerate}

\item Alice chooses random $m=2^{N+1}n$-bit strings $k, b$ where $i \in \{2, \cdots , N\}$, creates the state
$$
\begin{aligned}
|\psi \rangle &=\left(H^{b_{1}} \otimes \cdots \otimes H^{b_{m}}\right)|k\rangle \quad  \\
&=H^{(b)}|k\rangle
\end{aligned}
$$
and sends the state $|\psi \rangle$ to all the members.

\item Everyone chooses a random $m$-bit string $b^{i}$, receives a noisy version of $|\psi_i \rangle$, applies $H^{(b^i)}=\left(H^{b_{1}^{i}} \otimes \cdots \otimes H^{b_{2n}^{i}}\right)$ to $|\psi_i \rangle$, and measures the qubits on the computational basis to form the bit string $r^{i}$.

\item All of them publicly disclose $b$ and $b^{i}$. If $b_{j} \neq b_{j}^{i}$ they discard the $j^{\text {th }}$ bit of $k$ and $k^{i}$. They are left with around $2 n$ bits, say $v+e_i$.

\item They publicly decide on a permutation $P$ and perform $P$ on their $2 n$ bits.

\item They publicly disclose the second $n$ bits, to compute $w$ from $w^{(A)}$ and $w^{(i)}$ of their respective permuted $2 n$ bits. They also compute their respective bit flip error syndromes $s_{X}^{(A)}$ and $s_{X}^{(i)}$ on their remaining bits.

\item Based on some small mutually chosen confidence factor $c \geq 0$, everyone calculate $t=2 w t(w)+c n-1$.

\item They then decide on a $t^{\prime}$-error correcting $[n, k, d]$ binary linear code $\mathscr{C}_{1}$ for any $t^{\prime} \geq t$, such that $\mathscr{C}_{1}$ could be used in a CSS code. If no such binary linear code exists, then they abort. She also selects a random $u \in \mathscr{C}_{1}$ such that $k$ is in the coset of $u$.

\item The value of $v+u$ is announced by Alice. Here $u \in \mathscr{C}_{1}$ is a random codeword and $v$ represents the string of remaining non check qubits.

\item All sides subtracts the value of $v+u$ from their result to get $u+e_i$. Then they error correct to remove $e_i$ from their respective qubits to obtain a codewored $u \in \mathscr{C}_{1}$.

\item Alice and Bob use the coset of $u$ as the key.

\end{enumerate}

It was suggested that the parties should agree on a permutation to permute all the qubits prior to the error detection by Shor and Preskill \cite{shorpresskill}. However, this is unnecessary for the secure preparation and measurement of QKD, as two random permutations are equivalent to one random permutation.


As $\mathscr{C}_1$ is just a classical linear error-correcting code, step 9 serves as the error correction step. Privacy amplification is performed in step 10. Note that the value of $x_{k}$ is disclosed publicly, but the error-correcting code $\mathscr{C}_{1}$ chosen by Alice is random. Therefore from just $x_{k}$ it is not possible for Ever to know from which coset the $x_{k}$ belongs.

By proving that the entangled-based multiparty is secure, we have proved that our multiparty prepare and measure QKD is secure, using successive reduction to that scheme. Due to the fact that we have not made any alteration to the quantum system of Eve, we finally conclude that through the equivalency, multiparty prepare and measure is secure.

\section{Conclusion}

Quantum computers have opened up a new dimension not only for QKD but also in other cryptographic domains \cite{nilesh}. Greenberger Horne Zeilinger (GHZ) entanglement, in particular for multiparty quantum communication applications, is a helpful resource in this regard. However, actual applications of these multiparty activities have proven to be experimentally hard because to the low intensity and fragility of the GHZ entanglement source under present technology \cite{fu}. As a result, even with cutting-edge equipment, multiparty entangled-based quantum communication remains an extremely difficult experiment. The experimental distribution of GHZ entanglement has just recently been accomplished, with each pair of GHZ-entangled photons separated by less than 1km \cite{erven}. As a result, despite the fact that the multiparty entangled-based QKD was secure, it was not practical based on our current technology. Therefore the prepare and measure version of multiparty QKD provides a better alternative.

We extended the simple proof of security for the BB84 protocol to an arbitrary number of $N$ participants willing to share a secret key, achieving the so-called multiparty prepare and measure QKD. We have now established the security of the multiparty prepare-and-measure QKD by first establishing the security of its entanglement-based implementation and then gradually lowering it to the prepare-and-measure shceme. We made no changes on the quantum state of Eve, therefore we can conclude that the multiparty prepare-and-measure protocol is secure. There are, of course, some constraints. This proof establishes only the security of the ideal protocol with an ideal device, in which the states transmitted are identical to those specified. This does not, however, provide protection against attacks such as photon number splitting. Additionally, the proof makes no judgments regarding the amount of work required for decoding: for practical implementations in QKD, the code $\mathscr{C}_1$ must be decodable efficiently. Another issue is that, because CSS codes aren't ideal, the proof doesn't establish an upper bound for the amount of eavesdropping that is acceptable.

One should note that even though the entangled-based multiparty QKD is equivalent to the prepare and measure multiparty QKD, it does not implies that they are both equally feasible and practical with the same technology. However, it implies that the proof of security for the entangled-based protocol can be automatically reduced to the proof of security for the prepare and measure protocol. And we have used this technique to provide the simple proof of security for the multiparty prepare and measure QKD in this paper. This makes proving the security of the prepare and measure QKD very convenient because of the monogamy property of maximally entangled states used in the entangled-based version \cite{monogamy}.


\begin{thebibliography}{100}

\bibitem[1]{Advances} Pirandola, Stefano, et al. "Advances in quantum cryptography." Advances in optics and photonics 12.4 (2020): 1012-1236.tems (1982). doi: 10.1145/357172.357176.


\bibitem[2]{experiment} Zhao, Yi, et al. "Experimental quantum key distribution with decoy states." Physical review letters 96.7 (2006): 070502.

\bibitem[3]{sudhir} Singh, Sudhir Kumar, and R. Srikanth. "Unconditionally secure multipartite quantum key distribution." arXiv preprint quant-ph/0306118 (2003).

\bibitem[4]{QSS} Hillery, M., Bužek, V., and Berthiaume, A. (1999). Quantum secret sharing. Physical Review A, 59(3), 1829.

\bibitem[5]{QCC} Zhao, Shuai, et al. "Phase-matching quantum cryptographic conferencing." Physical Review Applied 14.2 (2020): 024010.

\bibitem[6]{TQC} Chen, Yu-Ao, et al. "Experimental quantum secret sharing and third-man quantum cryptography." Physical review letters 95.20 (2005): 200502.

\bibitem[7]{Xiao} Xiao, L., Long, G. L., Deng, F. G., and Pan, J. W. (2004). Efficient multiparty quantum-secret-sharing schemes. Physical Review A, 69(5), 052307.

\bibitem[8]{Epping} Epping, M., Kampermann, H., and Bruß, D. (2017). Multi-partite entanglement can speed up quantum key distribution in networks. New Journal of Physics, 19(9), 093012.

\bibitem[9]{Ribeiro} Ribeiro, J., Murta, G., and Wehner, S. (2018). Fully device-independent conference key agreement. Physical Review A, 97(2), 022307.

\bibitem[10]{GHZ} Greenberger, D. M., Horne, M. A., Shimony, A., and Zeilinger, A. (1990). Bell’s theorem without inequalities. American Journal of Physics, 58(12), 1131-1143.

\bibitem[11]{GHZ2} Pan, J. W., and Zeilinger, A. (1998). Greenberger-horne-zeilinger-state analyzer. Physical Review A, 57(3), 2208.

\bibitem[12]{assumptions2} Matsumoto, R. (2007). Multiparty quantum-key-distribution protocol without use of entanglement. Physical Review A, 76(6), 062316.

\bibitem[13]{assumptions1} Wallden, P., Dunjko, V., Kent, A., and Andersson, E. (2015). Quantum digital signatures with quantum-key-distribution components. Physical Review A, 91(4), 042304.

\bibitem[14]{BB84} H. Bennett, C. H. and Brassard, G. Quantum Cryptography: Public Key Distribution and Coin Tossing. Theoretical Computer Science 560 (1984). doi: 10.1016/j.tcs.2014. 05.025.

\bibitem[15]{shorpresskill} Shor, P. W., and Preskill, J. (2000). Simple proof of security of the BB84 quantum key distribution protocol. Physical review letters, 85(2), 441.

\bibitem[16]{mayers} Dominic Mayers. 2001. Unconditional security in quantum cryptography. J. ACM 48, 3 (May 2001), 351–406. DOI:https://doi.org/10.1145/382780.382781

\bibitem[17]{lochau} Lo, H. K., and Chau, H. F. (1999). Unconditional security of quantum key distribution over arbitrarily long distances. science, 283(5410), 2050-2056.

\bibitem[18]{BDSW} Bennett, C. H., DiVincenzo, D. P., Smolin, J. A., and Wootters, W. K. (1996). Mixed-state entanglement and quantum error correction. Physical Review A, 54(5), 3824.

\bibitem[19]{hamming} Gottesman, D. (1996). Class of quantum error-correcting codes saturating the quantum Hamming bound. Physical Review A, 54(3), 1862.

\bibitem[20]{mceliece} McEliece, R. J. (1977). The theory of information and coding: A mathematical framework for communication. Reading, Mass: Addison-Wesley Pub. Co., Advanced Book Program.

\bibitem[21]{CSS} Calderbank, A. R., and Shor, P. W. (1996). Good quantum error-correcting codes exist. Physical Review A, 54(2), 1098.

\bibitem[22]{CSS2} Steane, A. M. (1996). Simple quantum error-correcting codes. Physical Review A, 54(6), 4741.

\bibitem[23]{noclonning} Buzek, V., and Hillery, M. (1996). Quantum copying: Beyond the no-cloning theorem. Physical Review A, 54(3), 1844.

\bibitem[24]{dos} Li, Yuan, et al. "A denial-of-service attack on fiber-based continuous-variable quantum key distribution." Physics Letters A 382.45 (2018): 3253-3261.

\bibitem[25]{EDP} Rozpedek, Filip, et al. "Optimizing practical entanglement distillation." Physical Review A 97.6 (2018) 062333.

\bibitem[26]{knill} Knill, E., and Laflamme, R. (1997). Theory of quantum error-correcting codes. Physical Review A, 55(2), 900.

\bibitem[27]{cssstate} Chen, Kai, and Hoi-Kwong Lo. "Conference key agreement and quantum sharing of classical secrets with noisy GHZ states." Proceedings. International Symposium on Information Theory, 2005. ISIT 2005.. IEEE, 2005.

\bibitem[28]{twocolor} Aschauer, H., Dür, W., and Briegel, H. J. (2005). Multiparticle entanglement purification for two-colorable graph states. Physical Review A, 71(1), 012319.

\bibitem[29]{dur} Dür, W., Aschauer, H., and Briegel, H. J. (2003). Multiparticle entanglement purification for graph states. Physical review letters, 91(10), 107903.

\bibitem[30]{hostens1} Hostens, E., Dehaene, J., and De Moor, B. (2005). Stabilizer states and Clifford operations for systems of arbitrary dimensions and modular arithmetic. Physical Review A, 71(4), 042315.

\bibitem[31]{hostens2} Hostens, E., Dehaene, J., and De Moor, B. (2006). Hashing protocol for distilling multipartite Calderbank-Shor-Steane states. Physical Review A, 73(4), 042316.

\bibitem[32]{complexity} Brukner, Č., Żukowski, M., Pan, J. W., and Zeilinger, A. (2004). Bell’s inequalities and quantum communication complexity. Physical review letters, 92(12), 127901.

\bibitem[33]{nonlocality} Cereceda, J. L. (2004). Hardy's nonlocality for generalized n-particle GHZ states. Physics Letters A, 327(5-6), 433-437.

\bibitem[34]{lee} Lee, Juhui, et al. "Entanglement swapping secures multiparty quantum communication." Physical Review A 70.3 (2004): 032305.

\bibitem[35]{acin} Acin, Antonio, et al. "Classification of mixed three-qubit states." Physical Review Letters 87.4 (2001): 040401.

\bibitem[36]{normalform} Verstraete, F., Dehaene, J., and De Moor, B. (2003). Normal forms and entanglement measures for multipartite quantum states. Physical Review A, 68(1), 012103.

\bibitem[37]{cunha} M Cunha, M., Fonseca, A., and O Silva, E. (2019). Tripartite entanglement: Foundations and applications. Universe, 5(10), 209.

\bibitem[38]{E91} Ekert, Artur K. "Quantum Cryptography and Bell’s Theorem." Quantum Measurements in Optics. Springer, Boston, MA, 1992. 413-418.

\bibitem[39]{repeater} Jiang, Liang, et al. "Quantum repeater with encoding." Physical Review A 79.3 (2009): 032325.

\bibitem[40]{FTQC} Gottesman, Daniel. "Theory of fault-tolerant quantum computation." Physical Review A 57.1 (1998): 127.

\bibitem[41]{bone} de Bone, Sébastian, et al. "Protocols for creating and distilling multipartite GHZ states with Bell pairs." IEEE Transactions on Quantum Engineering 1 (2020): 1-10.

\bibitem[42]{nilesh} Nilesh, K., and Panigrahi, P. K. (2021). Quantum Blockchain based on Dimensional Lifting Generalized Gram-Schmidt Procedure. arXiv:2110.02763.

\bibitem[43]{fu} Fu, Yao, et al. "Long-distance measurement-device-independent multiparty quantum communication." Physical review letters 114.9 (2015): 090501.

\bibitem[44]{erven} Erven, Chris, et al. "Experimental three-photon quantum nonlocality under strict locality conditions." Nature photonics 8.4 (2014): 292-296.

\bibitem[45]{monogamy} Oliveira, T. R., Cornelio, M. F., and Fanchini, F. F. (2014). Monogamy of entanglement of formation. Physical Review A, 89(3), 034303.



\end{thebibliography}
\end{document}